\documentclass[a4paper,10pt]{article}

\usepackage[utf8]{inputenc}
\usepackage[margin=1in]{geometry}

\usepackage{amsmath,amsthm}
\usepackage{amssymb}

\usepackage{todonotes}

\usepackage{authblk}
\usepackage{hyperref}
\usepackage{subcaption}
\usepackage{enumitem}
\usepackage{cleveref}

\newcommand{\mc}{\mathcal}

\newcommand{\R}{\mathbb{R}}

\newcommand{\dx}[1]{\mathrm{d}#1}

\DeclareMathOperator*{\argsup}{\arg \sup}
\DeclareMathOperator*{\arginf}{\arg \inf}

\theoremstyle{plain}
\newtheorem{theorem}{Theorem}
\newtheorem{lemma}{Lemma}

\theoremstyle{definition}
\newtheorem{definition}{Definition}

\newtheorem{remark}{Remark}

\theoremstyle{remark}
\newtheorem*{claim*}{Claim}

\usepackage{tikzit}

\tikzstyle{Plain}=[fill=white, draw=black, shape=circle]
\tikzstyle{Square}=[fill=white, draw=black, shape=rectangle]

\tikzstyle{Both}=[draw=black, fill=none, <->]
\tikzstyle{Directed}=[->]
\tikzstyle{Directed snake}=[->, draw=black, decorate, decoration=snake]
\tikzstyle{Red Dashed}=[-, dashed, draw=red]
\tikzstyle{Red directed}=[->, draw=red]
\tikzstyle{Red directed snake}=[->, draw=red, decorate, decoration=snake]
\tikzstyle{Blue directed}=[->, draw=blue]
\tikzstyle{Dashed}=[-, dashed]
\tikzstyle{Red dashed directed}=[->, dashed, draw=red]

\title{Non-Obvious Manipulability for Single-Parameter Agents and Bilateral Trade}

\author{
    Thomas Archbold \quad
    Bart de Keijzer \quad
    Carmine Ventre
}

\date{February 15, 2022}

\affil{
    King's College London \\
    \texttt{
    \string{%
    \href{mailto:thomas.archbold@kcl.ac.uk}{thomas.archbold},%
    \href{mailto:bart.de_keijzer@kcl.ac.uk}{bart.de\_keijzer},%
    \href{mailto:carmine.ventre@kcl.ac.uk}{carmine.ventre}%
    \string}@kcl.ac.uk} \\
}

\begin{document}

\maketitle 

\begin{abstract}
    A recent line of work in mechanism design has focused on guaranteeing incentive compatibility for agents without  contingent reasoning skills: \emph{obviously strategyproof  mechanisms} guarantee that it is ``obvious'' for these imperfectly rational agents to behave honestly, whereas \emph{non-obviously manipulable} (NOM) mechanisms take a more optimistic view and assume that these agents will only misbehave when it is ``obvious'' for them to do so.
    Technically, obviousness requires comparing certain extrema (defined over the actions of the other agents) of an agent's utilities for honest behaviour against dishonest behaviour. 

    We present a technique for designing NOM mechanisms in settings with monetary transfers based on cycle monotonicity, which allows us to disentangle the specification of the mechanism's allocation from its payments.
    By leveraging this framework, we completely characterise both allocation and payment functions of NOM mechanisms for single-parameter agents.
    We then look at the classical setting of bilateral trade and study how much subsidy, if any, is needed to guarantee NOM, efficiency, and individual rationality.
    We prove a stark dichotomy: no finite subsidy suffices if agents look only at best-case extremes, whereas no subsidy at all is required when agents focus on worst-case extremes.
    We conclude the paper by characterising the NOM mechanisms that require no subsidies whilst satisfying individual rationality.
\end{abstract}

\section{Introduction}

The prevailing solution concept in mechanism design is dominant-strategy incentive-compatibility, otherwise known as strategyproofness, which stipulates that for any joint strategy of the other agents a given agent always weakly prefers to act truthfully than dishonestly.
A classic example of a strategyproof mechanism is the second-price sealed-bid auction, in which it is always in an agent's best interests to submit her true valuation for the item to the auctioneer.
While this is appealing from a theoretical standpoint it relies on the assumption that the players of the game act perfectly rationally, meaning they are able to correctly reason about the behaviour of the other players and solve the resulting optimisation problem to identify an optimal strategy.
This may place a high cognitive or temporal burden on the player and it can be unrealistic to assume this is how they act in the real world.

Empirical evidence in fact shows that people facing mechanisms that are not strategyproof might fail to understand that lying is beneficial.
Examples include mechanisms both with and without monetary transfers.
The Deferred Acceptance Algorithm, used in many countries to allocate doctors to hospitals, is known to perform well in practice and produce stable outcomes, even though hospitals could lie and hire better candidates \cite{Roth1991}.
Uniform-price auctions are a natural generalisation of second-price auctions to the multi-unit case: with a supply of $k$, they charge the winners the $(k+1)$-th highest bid for each unit received.
Variants of uniform-price auctions are used by search engines to charge for sponsored ad slots and governments to sell bonds.
These auctions are famously not strategyproof and yet there is little evidence that bidders cheat in practice \cite{Kastl2011}.
On the other hand, there are manipulable mechanisms where people can easily figure out how to lie profitably: the Boston Mechanism for school choice is known to be manipulable \cite{Pathak2008} and empirical data shows evidence of misreports \cite{Dur2018}.
Similarly, there is evidence of bid shading in pay-as-bid multi-unit auctions \cite{Hortacsu2010}. 

Motivated by these observations, \cite{Troyan2020} study \emph{non-obvious manipulability}, an alternative solution concept to strategyproofness for cognitively limited agents.
Under this relaxation, simply referred to as NOM, a mechanism is allowed to be vulnerable to certain types of manipulations which are harder to identify by agents with bounded rationality.
Instead of requiring the truthful action for each agent $i$ to be a dominant strategy at each possible interaction with the mechanism, under non-obvious manipulability we need only to compare best-case utilities for truthful outcomes (i.e., the maximum utility, taken over the reports of the other agents, for $i$ when she is truthful) to best-case utilities for dishonest outcomes, and worst-case utilities for truthful outcomes to worst-case utilities for dishonest outcomes.
This is justified under the assumption that cognitively-limited agents can only identify this type of \emph{obvious} manipulation where the benefits occur at the extremes of her utility function.

In this paper we study NOM mechanisms that use monetary transfers.
We wish to: (i) understand the restrictions NOM imposes on the social choice function; (ii) characterise both the allocation functions and payments in NOM mechanisms for the large class of single-parameter agents; and (iii) assess the power of NOM mechanisms for two-sided markets.
We develop both conceptual and technical tools to obtain answers to each of these desiderata.

\bigbreak

\textbf{Our Contribution.}
Our goal is to understand how permissive NOM is as a solution concept for allocation functions.
NOM requires truthtelling to yield a greater utility than misreporting for every player under the best and worst outcomes of the mechanism, taken over all actions of the other players.
Therefore the constraints imposed by NOM \emph{depend on} both the allocation and payment function since they define these maximum and minimum utilities, and the utility of every other outcome must lie between these extremes.
The first contribution of this work is to abstract this dependency using what we call \emph{labellings} and use these to extend the cycle monotonicity technique to handle NOM direct-revelation mechanisms.
Informally a labelling is a choice of bid profiles designating the best and worst outcomes for a given player that the mechanism will define through its allocation and payments.
These induce two types of constraints: incentive-compatibility constraints, which enforce the NOM property where the extreme utilities when lying are no better than the extreme utilities from truthtelling, and labelling-induced constraints, which ensure the validity of the labels.
Given an allocation function $f$ and labelling $\lambda$ for a given agent $i$ we can then define a graph $\mc{G}_{i,f}^\lambda$ representing the constraints imposed by the mechanism and show that the existence of payments satisfying these constraints is related to cycles in this graph.

\begin{center}
    \begin{minipage}[c]{0.9\linewidth}
        \textbf{Main Result 1 (informal).} \textit{An allocation function $f$ is NOM implementable if and only if there exists a labelling $\lambda$ such that $\mc{G}_{i,f}^\lambda$ has no negative-weight cycles.}
    \end{minipage}
\end{center}
This provides a useful tool for studying the allocation functions for which we can define payments leading to NOM mechanisms.
We apply this to the well-studied setting of single-parameter agents where each agent's type can be expressed as a single number, and their utility for an allocation is simply the product of their type with the amount they are allocated.
Negative cycles exist in the graph $\mc{G}_{i,f}^\lambda$ when the labels are antimonotone, roughly meaning that for two types $b_i$ and $b_i'$ with $b_i < b_i'$ the agent is allocated more for the label corresponding to $b_i$ than it does for the label corresponding to $b_i'$.
We then prove that a monotone labelling guarantees no negative cycles and that such a labelling is possible whenever the image of $f_i$ is rich enough.
Roughly speaking $f$ is \emph{overlapping} if for each agent $i$ and pair of types $b_i, b_i'$ the sets $\{ f_i(b_i, b_{-i}) \}_{b_{-i}}$ and $\{ f_i(b_i', b_{-i}) \}_{b_{-i}}$ have a non-empty intersection.

\begin{center}
    \begin{minipage}[c]{0.9\linewidth}
        \textbf{Main Result 2 (informal).} \textit{An allocation function $f$ is NOM-implementable for single-parameter agents if and only if $f$ is overlapping.}
    \end{minipage}
\end{center}
Our second contribution proves that as long as the labels define a a monotone restriction of the allocation function then it is possible to define a NOM payment scheme.
This uncovers a connection with strategyproofness for single-parameter agents, where the allocation must be monotone for each fixed $b_{-i}$, and when this is the case it is known from \cite{ArcherTardos2001} the form the payments must take.
We observe that we can take the same payments along the monotone restriction defined by the labels, while we can easily characterise the payments associated with the remaining bid profiles.
This gives a complete picture on the power of NOM mechanisms for single-parameter agents.

We conclude by applying our findings to NOM mechanisms for two-sided markets where both sides of the market are strategic, a setting which is widespread but notoriously too complex for strategyproofness.
We focus on bilateral trade where a single item may be sold from a seller to a buyer.
We want to guarantee efficiency, where trade occurs only when the buyer's valuation is at least the seller's production cost; weak budget balance (WBB), for which we never pay the seller more than we charge the buyer; and individual rationality (IR), where the utilities of both buyer and seller are non-negative.
An impossibility result from \cite{Troyan2020} tells us that NOM cannot guarantee these three properties simultaneously, so we investigate the extent to which a multiplicative subsidy of $\alpha$, which guarantees that we never pay the seller more than an $\alpha \ge 1$ factor of what we receive from the buyer (i.e., $\alpha$-WBB), allows us to obtain NOM, efficiency, and IR simultaneously.
\begin{center}
    \begin{minipage}[c]{0.9\linewidth}
        \textbf{Main Result 3 (informal).} \textit{Any NOM, IR and $\alpha$-WBB mechanism for bilateral trade has unbounded $\alpha$.}
    \end{minipage}
\end{center}
The approach developed in the first two main theorems can be applied to characterise BNOM and WNOM independently.
For our third main result we first characterise the monotone restrictions of a given allocation function that will guarantee efficiency, IR, and either BNOM or WNOM.
For the former the unbounded $\alpha$ is implied by the incentive-compatibility constraints: using the explicit payment formulae we show that the buyer must receive the item for free while the seller must receive some positive payment in each of their respective best case profiles.
For the latter it is implied by the labelling validity constraints.

We then look beyond using monotone restrictions of the allocation function to construct NOM mechanisms where the labelled profiles need not appear on a ``single line''.
Can we be more flexible in our labelling scheme to derive different payments that would satisfy IR and finite subsidies for an allocation function that is overlapping and efficient?
The answer is surprising: even non-single line BNOM mechanisms require infinite subsidies while non-single line WNOM mechanisms can guarantee WBB (i.e., no subsidies).
There exists in fact a definition of WNOM labels that implement first-price auctions.
We may also strengthen the notion of incentive-compatibility on either side of the market to strategyproofness -- for example, the buyer is strategyproof and the seller WNOM -- by combining a first-price auction on one side with a second-price auction on the other.
This leaves the spread to the market designer, whereas a fully WNOM market would be budget balanced.
This setting may be of interest for some classes of double-sided markets where the rationality is asymmetric and depends on the role played in the market.

We conclude by characterising individually rational NOM mechanisms that might be inefficient but require no subsidies.
Intuitively, this result says that there exist trading windows defined by payment thresholds that partition the domain of buyer and seller into three sets; bids guaranteeing that the trade occurs, bids guaranteeing that the trade will not occur, and bids where the decision to trade depends on the bid at the other side of the market.

\bigbreak

\textbf{Related Work.}
There has been a growing body of work in the algorithmic mechanism design literature that studies agents with bounded rationality.
\cite{Li2017} introduced the aforementioned notion of obvious strategproofness (OSP) to study settings in which agents lack contingent reasoning skills, and provides two characterisations of such mechanisms, one from the perspective of contingent reasoning and the other from that of commitment power.
While OSP mechanisms provide stronger incentive guarantees than strategyproof ones, they are harder to construct.
Where the revelation principle allows us to restrict our attention to direct mechanisms for strategyproofness, the same is not true for OSP: the extensive-form implementation of the mechanism is crucial to a mechanism's obvious strategyproofness.
The relationship between the extensive-form implementation and the allocation function of OSP mechanisms is studied in \cite{Ferraioli2018ESA}, where together with the introduction of the cycle monotonicity variant for OSP mechanisms with transfers, the authors give bounds on the approximation guarantee of these mechanisms for scheduling related machines.
Binary allocation problems are then studied in \cite{Ferraioli2019} first for single-parameter agents with small domains and more recently in \cite{Ferraioli2022WINE} for the general case.

OSP mechanisms in which the mechanism designer possesses some form of verification power are considered in  \cite{Ferraioli2021}, \cite{Ferraioli2018IJCAI}, and \cite{Kyropoulou2019}.
The work of \cite{deKeijzer2020} explores OSP in the multiparameter setting, specifically OSP combinatorial auctions with single-minded bidders where each agent's desired bundle is unknown.
A relaxation of OSP interpolating between the lack of contingent reasoning modelled by OSP and the full rationality of strategyproofness is introduced in \cite{Ferraioli2022TOCS}, where they show that increasing the ``degree of rationality'' of the agents only slightly improves the approximation guarantees of incentive-compatible mechanisms for machine scheduling and facility location.

\cite{Troyan2020} consider agents with bounded rationality from the opposite perspective and introduce non-obvious manipulability (NOM) to study the types of manipulations available to a cognitively limited agent.
In contrast to OSP, where the worst-case truthful outcome must be at least as good as the best-case dishonest outcome, NOM stipulates that both the best- and worst-case truthful outcomes are no worse than the best- and worst-case dishonest outcomes, respectively.
They provide a characterisation for direct-revelation NOM mechanisms and classify existing mechanisms as either obviously manipulable (OM) or NOM in the settings of school choice, two-sided matching, and auctions.
They show that a manipulation is obvious if and only if it can be recognised by a cognitively limited agent.
They go on to characterise the class of NOM matching mechanisms for settings subject to both one-sided and two-sided manipulations, then show that for generalisations of first-price and second-price auctions, the former is OM while the latter is NOM.
Finally they show that every efficient, IR, WBB mechanism for bilateral trade is OM.

Several papers build directly on Troyan and Morrill's work.
\cite{Aziz2020} study NOM voting rules and provide conditions under which certain classes of rules are NOM, and this may depend on the number of outcomes relative to the number of voters.
They also provide algorithms to compute such manipulations using polynomial-time reductions to the unweighted coalitional manipulation problem, which yield polynomial-time algorithms under the $k$-approval voting rule.
\cite{Ortega2019} study indirect NOM mechanisms for cake-cutting and show that relaxing strategyproofness to NOM resolves the conflict between truthtelling and fairness: unlike strategyproofness, NOM is compatible with proportionality.
Like \cite{Li2017} they note how incentive properties may vary between theoretically equivalent mechanisms for agents with bounded rationality.
\cite{Psomas2022} study fairly allocating indivisible goods to agents with additive valuations.
They show the existence of deterministic NOM mechanisms which achieve envy-freeness up to one good (EF1), and highlight a conflict for NOM under different objective functions -- they provide a social welfare maximising NOM mechanism for $n \ge 3$ agents and in contrast show that any optimal mechanism for egalitarian or Nash welfare is obviously manipulable.
They also provide an efficiency-preserving black box reduction from designing NOM and EF1 mechanisms to designing EF1 algorithms.

\section{Preliminaries}

\label{sec:preliminaries}

For a natural number $k$ we denote by $[k]$ the set $\{ 1, 2, \ldots, k \}$.
We consider mechanism design settings with the possibility of monetary transfers.
There is a set of agents $[n]$, where each agent $i$ has some type $t_i$ that comes from a set $D_i$ of possible types (also called agent $i$'s domain).
We let $D = \times_{i \in [n]} D_i$ be the set of type profiles.
For a set of possible outcomes $O$ an agent's type describes the utility she receives from each outcome $o \in O$.
We can therefore think of player $i$'s type as a function $t_i : O \to \mathbb{R}$ mapping outcomes to numbers.
For a type profile $t = (t_1,t_2,\ldots,t_n)$, player $i$, and type $b_i \in D_i$, we use the standard notation $(b_i,t_{-i})$ to denote the type profile obtained from $t$ by replacing $t_i$ with $b_i$.
Furthermore we define $D_{-i} = \times_{j \in [n] \setminus \{i\}} D_j$.

A mechanism collects a type profile from the agents and uses this to return an outcome from $O$ and a vector of payments.
We focus on direct-revelation mechanisms, as in \cite{Troyan2020}, whereby agents report their types directly to mechanism.
We view a mechanism as a function $M : D \to O \times \R^n$ mapping type profiles to outcome-payment vector pairs, where $p_i$ describes the payment made to agent $i$.
Agents are assumed to maximise their \emph{utility} and are equipped with standard quasi-linear utility functions so that, on output $(o,p)$ of the mechanism, agent $i$ derives utility $t_i(o) + p_i$.
We overload notation and denote the utility of agent $i$ under this output as $t_i(o,p)$.
A \emph{social choice function} $f : D \to O$ is a function that maps type profiles to outcomes, and we say that a mechanism $M$ \emph{implements} a social choice function $f$ if and only if for all $t \in D$ the outcome returned by $M(t)$ is equal to $f(t)$.
Note that agent types are private information, hence each agent $i$ may misreport her type as $b_i \neq t_i$ to the mechanism.
We will therefore refer to the reports as \emph{bids} and the profile submitted to the mechanism as a \emph{bid profile}.
We use payments to realign the agents' incentives to satisfy the NOM property (defined precisely later in this section).
In this paper we focus on designing NOM mechanisms that implement a given social choice function and therefore, given a social choice function $f$, we may view our mechanism as the tuple $M = (f,p)$ where the outcome returned by $M$ on input $b$ is equal to $f(b)$ for all bid profiles $b$, and $p : D \to \mathbb{R}^n$ is a rule defining the payments made to each agent for each bid profile $b$.
Therefore our objective, given $f$, is to find payment functions $p$ such that the resulting mechanism $M = (f,p)$ is NOM.

From \Cref{sec:single-parameter} we focus on a specific class of social choice functions known as \emph{allocation functions}.
Here we assume outcomes to be real-numbered vectors of length $n$ and refer to them as \emph{allocations}, and assume an agent's utility for an allocation $a$ depends only on the $i$th coordinate of $a$.
This is done purely for interpretational purposes, since we focus on auction settings in which the mechanism assigns an amount of some object to each agent and an agent's type describes the utility she gets from that allocation.

A familiar solution concept in mechanism design is strategyproofness: a mechanism $M$ is \emph{strategyproof} if and only if $t_i(M(t_i,b_{-i})) \ge t_i(M(b_i,b_{-i}))$ for all $t_i,b_i \in D_i$, all $b_{-i} \in D_{-i}$, and all players $i$.
Non-obvious manipulability, on the other hand, requires agents only to compare the extremes of their utility function.
Informally stated it says that a manipulation is \emph{obvious} if either the best-case dishonest outcome is strictly greater than the best-case truthful outcome or the worst-case dishonest outcome is strictly greater than the worst-case truthful outcome.
The best and worst cases are defined over all bid profiles of the other players excluding $i$.
As per \cite{Troyan2020} a direct-revelation mechanism is \emph{not obviously manipulable (NOM)} if both of the following two properties hold:
\begin{gather}
    \label{eq:direct-revelation-BNOM}
    \sup_{b_{-i}} \{ t_i(M(t_i,b_{-i})) \} \ge \sup_{b_{-i}} \{ t_i(M(b_i,b_{-i})) \}, \\
    \label{eq:direct-revelation-WNOM}
    \inf_{b_{-i}} \{ t_i(M(t_i,b_{-i})) \} \ge \inf_{b_{-i}} \{ t_i(M(b_i,b_{-i})) \}
\end{gather}
for every $t_i,b_i \in D_i$ and for every player $i$.
Note that the $b_{-i}$ are not necessarily the same on each side of the inequalities.
If \eqref{eq:direct-revelation-BNOM} holds then $M$ is \emph{best-case not obviously manipulable (BNOM)} and if \eqref{eq:direct-revelation-WNOM} holds then $M$ is \emph{worst-case not obviously manipulable (WNOM)}.
If either inequality is violated for some $t_i,b_i \in D_i$ then $b_i$ is an \emph{obvious manipulation} of $M$.
If $M = (f,p)$ implements $f$ and we can define payments $p$ such that $M$ is BNOM (respectively, WNOM) then we say that $f$ is \emph{BNOM-implementable} (respectively, \emph{WNOM-implementable}), and if $f$ is both BNOM- and WNOM-implementable then $f$ is \emph{NOM-implementable}.

In \Cref{sec:single-parameter,sec:bilateral-trade} we consider several additional properties.
A mechanism $M$ is \emph{individually rational (IR)} if $t_i(M(t_i,b_{-i})) \ge 0$ for all $t_i \in D_i$, all $b_{-i} \in D_{-i}$, and for each player $i$, and it makes \emph{no positive transfers (NPT)} if $p_i(b) \le 0$ for each player $i$ and each $b \in D$.
In \Cref{sec:bilateral-trade} we also study efficiency and budget-balance in the context of bilateral trade, which we will define when needed.

In order to refer to specific types of some player $i$ we enumerate her domain as $D_i = \{ t_1, t_2, \ldots, t_d \}$.
In \Cref{sec:single-parameter} we focus on single-parameter agents where player types are treated as scalars and we assume that $t_1 < t_2 < \ldots < t_d$.
We conclude this section with the following remark.
%
\begin{remark}
    \label{thm:SP-implies-NOM}
    If $M$ is strategyproof then it is NOM, however the reverse is not necessarily true.
\end{remark}

\begin{proof}
    We will prove the claim by showing that if $M$ is strategyproof then it is both BNOM and WNOM, starting with the former.
    Fix player $i$ with type $t_i$.
	By strategyproofness we have $t_i(M(t_i,b_{-i})) \ge t_i(M(b_i,b_{-i}))$ for all $t_i,b_i \in D_i$ and for all $b_{-i}$.
	Assume for contradiction that $M$ is obviously manipulable in the best case and therefore $\sup_{b_{-i}} \{ t_i(M(b_i,b_{-i})) \} > \sup_{b_{-i}} \{ t_i(M(t_i,b_{-i})) \}$ for some manipulation $b_i$.
	Let $\hat{x}_{-i} \in \arg \sup \{ t_i(M(b_i,b_{-i}))$ denote the bid profile leading to $i$'s best-case dishonest utility.
	Then
    \begin{equation*}
        t_i(M(t_i,\hat{b}_{-i})) \ge
        t_i(M(b_i,\hat{b}_{-i})) >
        \sup_{b_{-i}} \{ t_i(M(t_i,b_{-i})) \}
    \end{equation*}
    and thus $t_i(M(t_i,\hat{b}_{-i}))$ is strictly greater than the maximum utility $i$ gets when bidding $t_i$ truthfully, a contradiction.
    A similar argument shows that $M$ is WNOM.
    Let $\hat{b}_{-i} \in \arg \inf \{ t_i(M(b_i,b_{-i})) \}$ be the bid profile leading to $i$'s worst-case dishonest utility.
    So
    \begin{equation*}
        \inf_{b_{-i}} \{ t_i(M(b_i,b_{-i})) \} >
        t_i(M(t_i,\hat{b}_{-i})) \ge
        t_i(M(x_i,\hat{b}_{-i}))
    \end{equation*}
    meaning $t_i(M(b_i,\hat{b}_{-i}))$ is strictly less than the worst-case dishonest utility of player $i$.
    Thus strategyproof $M$ is also NOM.
\end{proof}
    
This says that there is some additional flexibility that we may be able to exploit when designing incentive-compatible mechanisms for NOM versus strategyproofness.
In the following section we formalise this flexibility in the form of profile labellings and use it to derive a useful technique for designing NOM mechanisms.

\section{Profile Labellings}

\label{sec:labellings}

While strategyproofness requires satisfying a constraint on each pair of bid profiles $(t_i,b_{-i})$ and $(b_i,b_{-i})$ for each pair of types $t_i,b_i \in D_i$ and each $b_{-i} \in D_{-i}$, for NOM we need only compare the extremes of $i$'s utility function.
Since the social choice function $f$ is given then we may design the payments $p$ for some mechanism $M=(f,p)$ to define these extremes and thus effectively select which bid profiles to compare for incentive-compatibility.
We then only need to satisfy the respective BNOM and WNOM constraints for these chosen profiles, in addition to the implicit constraints imposed by this ordering on the profiles.
We model this with \emph{profile labellings}, which allow us to designate the type profiles leading to an agent's highest and lowest utilities when interacting with the mechanism, based on her true type and the type she reports.


\begin{definition}[Profile labelling]
    Fix player $i$ with domain $D_i$ and let $d = |D_i|$. A \emph{best-case labelling $\beta$} and \emph{worst-case labelling $\omega$} of mechanism $M$ for player $i$ are matrices $\beta, \omega \in |D_{-i}|^{d \times d}$ such that
    \begin{gather*}
        \beta_{jk} \in \argsup_{b_{-i} \in D_{-i}} \{ t_j(M(t_k,b_{-i})) \}, \\
        \omega_{jk} \in \arginf_{b_{-i} \in D_{-i}} \{ t_j(M(t_k,b_{-i})) \}.
    \end{gather*}
\end{definition}
The entry in row $j$ and column $k$ of $\beta$ (respectively, $\omega$) therefore represents the partial bid profile that, when submitted to mechanism $M$ along with player $i$'s bid $t_k$, results in $i$'s greatest (respectively, least) utility when she has type $t_j$.
Often we will use $\lambda$ to denote either a best- or worst-case labelling of some mechanism.
For a labelled profile $\lambda_{jk}$ for player $i$, given that we can infer $i$'s bid by looking at the second subscript we will use $\lambda_{jk}$ to refer to the (full) bid profile $(t_k,\lambda_{jk})$ for brevity.
When dealing with labellings for multiple players as in \Cref{sec:bilateral-trade} we will use a superscript to differentiate between the two (i.e., $\lambda^i$ and $\lambda^j$ for players $i$ and $j$, respectively).

%

Profile labellings induce two types of constraint.
The \emph{labelling constraints} induced by $\lambda \in \{ \beta, \omega \}$ ensure that a labelled bid profile leads to the extreme values of $i$'s utility function when she has type $t_j$ and bids type $t_k$:
\begin{align}
    \label{eq:best-labelling-constraints}
    t_j(M(\beta_{jk})) & \ge t_j(M(t_k,b_{-i})) \quad \text{for all $b_{-i} \in D_{-i}$}, \\
    \label{eq:worst-labelling-constraints}
    t_j(M(\omega_{jk})) & \le t_j(M(t_k,b_{-i})) \quad \text{for all $b_{-i} \in D_{-i}$},
\end{align}
for all $t_j,t_k \in D_i$.
Meanwhile the \emph{incentive-compatibility constraints} ensure that the extreme values of $i$'s utility function when she reports truthfully to the mechanism are no worse than the extreme values of her utility function when she misreports her type.
Since the label $\lambda_{jk}$ denotes a bid profile where $i$ has type $t_j$ and bids type $t_k$ then for both best- and worst-case labellings we have:
\begin{equation}
    \label{eq:ic-constraints}
    t_j(M(\lambda_{jj})) \ge t_j(M(\lambda_{jk})) \quad \text{for all $t_j,t_k \in D_i$}.
\end{equation}
It is worth emphasising that the social choice function $f$ is fixed and that we will use the profile labellings to decide whether payments exist such that the resulting mechanism is NOM.
Not every labelling will admit such payments since the above constraints may be unsatisfiable.

We use the labellings to define a graph in order to characterise the class of NOM-implementable social choice functions.
This is achieved using the cyclic monotonicity technique in a similar manner to \cite{Rochet1987}, \cite{Lavi2009}, and \cite{Ventre2014}, among others.
For some social choice function $f$ with labelling $\lambda \in \{ \beta, \omega \}$ for player $i$ we construct the weighted directed multigraph $\mc{G}_{i,f}^\lambda = (V, E^\lambda, w)$ whose node set $V = D$ is the domain of $f$ and whose edge set encodes the constraints imposed by a NOM mechanism for $f$.
The precise edges in $E^\lambda$ will vary depending on whether $\lambda$ is a best- or worst-case labelling.
We annotate each edge of the graph with a type from $i$'s domain and for two nodes $x$ and $y$ connected by an edge annotated with type $t$ we write $(x,y;t)$, or equivalently $x \rightarrow^t y$.
The weight of this edge is $w(x,y;t) = t(f(x)) - t(f(y))$.
The weight $w(C)$ of a cycle $C$ is simply the sum of edge weights on each edge in the cycle.

The edges of the graph represent the constraints in \eqref{eq:best-labelling-constraints}, \eqref{eq:worst-labelling-constraints}, and \eqref{eq:ic-constraints}.
For any labelling $\lambda$ the edge set $E^\lambda$ contains the edges $\{ \lambda_{jj} \rightarrow^{t_j} \lambda_{jk} \, : \, t_j,t_k \in D_i \}$ representing the incentive-compatibility constraints.
The labelling constraints for $\beta$ and $\omega$ are the reverse of each other and we have $\{ \beta_{jk} \rightarrow^{t_j} (t_k,b_{-i}) \, : \, t_j,t_k \in D_i, b_{-i} \in D_{-i} \} \subset E^\beta$ and $\{ (t_k,b_{-i}) \rightarrow^{t_j} \omega_{jk} \, : \, t_j,t_k \in D_i, b_{-i} \in D_{-i} \} \subset E^\omega$.


For any labelling $\lambda$ of $f$ for player $i$, given some bid $t_k \in D_i$ every edge encoding a labelling constraint is contained within the subgraph of $\mc{G}_{i,f}^\lambda$ that consists only of nodes $(t_k,b_{-i})$, that is, profiles where $i$ bids $t_k$.
We refer to this subgraph as the ``$k$-island'' of $\mc{G}_{i,f}^\lambda$.
Thus all labelling constraint edges for profile $\lambda_{jk}$ are contained in the $k$-island of $\mc{G}_{i,f}^\lambda$, while incentive-compatibility constraint edges are those which for any two types $t_j, t_k \in D_i$ cross from the $j$-island to the $k$-island of $\mc{G}_{i,f}^\lambda$.
This graph allows us to use the cycle monotonicity technique to characterise the social choice functions that are NOM-implementable.

\begin{theorem}
    \label{thm:nom-implementable}
    Social choice function $f$ is BNOM-implementable (respectively, WNOM-implementable) if and only if for each agent $i$ there exists a best-case labelling $\beta$ (respectively, worst-case labelling $\omega$) such that $\mc{G}_{i,f}^\beta$ (respectively, $\mc{G}_{i,f}^\omega$) contains no negative weight cycles.
\end{theorem}

\begin{proof}
    We will prove the claim for WNOM and note that the proof for BNOM follows identical reasoning.
    
    ($\implies$) Suppose payments $p$ exist such that $(f,p)$ is WNOM. Fix player $i$ and assume for contradiction that every choice of $\omega$ induces a negative weight cycle in $G_{i,f}^\omega$. Let $C = \langle x^{(1)} \rightarrow^{t_1} x^{(2)} \rightarrow^{t_2} \ldots \rightarrow^{t_{k-1}} x^{(k)} \rightarrow^{t_k} x^{(1)} \rangle$ be such a cycle with $w(C) = w(x^{(1)},x^{(2)};t_1) + \ldots + w(x^{(k)},x^{(1)};t_k) < 0$. Since $f$ is WNOM-implementable then the inequality encoded by each edge is satisfied:
    \begin{gather*}
    t_1(f(x^{(1)}))+p_i(x^{(1)}) \ge t_1(f(x^{(2)}))+p_i(x^{(2)}) \\
    t_2(f(x^{(2)}))+p_i(x^{(2)}) \ge t_2(f(x^{(3)}))+p_i(x^{(3)}) \\
    \vdots \\
    t_k(f(x^{(k)}))+p_i(x^{(k)}) \ge t_k(f(x^{(1)}))+p_i(x^{(1)}) 
    \end{gather*}
    The sum of all the inequalities must also be satisfied, in which case the payments cancel out and we are left with:
    \[
    t_1(f(x^{(1)}))-t_1(f(x^{(2)})) + 
    t_2(f(x^{(2)}))-t_2(f(x^{(3)})) + 
    \ldots +
    t_k(f(x^{(k)}))-t_k(f(x^{(1)})) \ge 0,
    \]
    which is simply the expression for the sum of edge weights along $C$. Therefore we have $w(C) \ge 0$, a contradiction, so there must be a labelling $\omega$ such that $G_{i,f}^\omega$ has no negative weight cycles.

    ($\impliedby$) Let $\omega$ be a labelling such that $G_{i,f}^\omega$ has no negative weight cycles. Augment $G_{i,f}^\omega$ with an artificial node $\gamma$ and an edge $(\gamma,x;0)$ for each $x \in D$ with weight 0. Since $G_{i,f}^\omega$ contains no negative cycles and the addition of $\gamma$ creates no new cycles (since we only have outgoing edges from $\gamma$) then shortest paths are well-defined in this new graph. We will now set $p_i(x) = \text{SP}(\gamma,x)$, where $\text{SP}$ denotes the shortest path between two nodes in $G_{i,f}^\omega$. For any node $y$ connected to $x$ by an edge labelled by type $t$, by definition of the shortest path we have $\text{SP}(\gamma,y) \le \text{SP}(\gamma,x) + w(x,y;t)$. Thus $p_i(y) \le pi(x) + t(f(x)) - t(f(y))$ and so $t(f(x)) + p_i(x) \ge t(f(y)) + p_i(y)$ for any two profiles $x,y$ connected by an edge labelled by type $t$. It remains to show that, whenever such an edge appears as per our definition of $G_{i,f}^\omega$, the corresponding WNOM constraint that it represents is satisfied.
    
    \emph{Case 1: $t = x_i$.} In this case the edge $(x,y;t)$ either encodes an incentive-compatibility constraint or a truthful labelling constraint. Suppose first that it is the former. Then $x$ must be the profile in which $i$ gets her lowest utility when she bids $x_i$ truthfully, while $y$ must be the profile in which she gets her lowest utility when she has type $x_i$ and bids $y_i$ dishonestly. Thus $x_i(f(x)) + p_i(x) \ge x_i(f(y)) + p_i(y)$ and so her worst-case truthful outcome is no worse than her worst-case dishonest outcome when she has type $x_i$. Now suppose the edge encodes the latter constraint. Then we have $x=(x_i,x_{-i})$ and $y=(x_i,x_{-i})$, where $y$ designates the input yielding her her worst outcome when she has type $x_i$. Thus $i$ prefers $M(x)$ to $M(y)$ and we have $x_i(f(x)) + p_i(x) \ge x_i(f(y)) + p_i(y)$.
    
    \emph{Case 2: $t \neq x_i$.} In this case the edge $(x,y;t)$ encodes a labelling constraint for a dishonest interaction of $i$ with $M$. Now $y=(x_i,y_{-i})$ must be the profile in which $i$ receives her lowest utility when she has type $t$ and bids $x_i$ dishonestly. Thus she prefers the outcome when $M$ is given input $x=(x_i,x_{-i})$ and so $t(f(x)) + p_i(x) \ge t(f(y)) + p_i(y)$.
    
    In all cases whenever we have an edge $(x,y;t)$ all the constraints imposed by WNOM are satisfied.
    Therefore the mechanism $(f,p)$ is WNOM.
\end{proof}

We note that the result holds only for finite domains.
Certain cycles are guaranteed to exist in $\mc{G}_{i,f}^\beta$ and $\mc{G}_{i,f}^\omega$ regardless of the specific nodes selected as the labelled profiles.
These cycles fall into on of two types: cycles between edges which are entirely contained in a given $j$-island of $\mc{G}_{i,f}^\lambda$, and those that use edges that traverse different islands.
Intuitively the former class of cycles appears because given a bid $t_j$ of player $i$, each labelled profile for this bid and some true type from $i$'s domain must yield a higher (respectively, lower) utility from the mechanism than every other profile where $i$ bids $t_j$, including other profiles in the $j$-island labelled with a different true type.

We can select any number of true types from $D_i$ and form a cycle between their corresponding labelled profiles in the $j$-island.
Take any $p$ types $t_1,t_2,\ldots,t_p \in D_i$. 
Then no matter the specific labellings $\beta$ and $\omega$ there always exists the $p$-cycles
\begin{gather}
    \label{eq:bnom-guaranteed-island-cycles}
    C_p^\beta = \{  \beta_{k,j} \rightarrow^{t_k} \beta_{k+1,j}\, : \, k \in [p] \}, \\
    \label{eq:wnom-guaranteed-island-cycles}
    C_p^\omega = \{  \omega_{k-1,j} \rightarrow^{t_k} \omega_{k,j} \, : \, k \in [p] \},
\end{gather}
where the indices $-1$ and $p+1$ wrap around to $p$ and $1$ respectively.
Similarly for any $p$ types we can form a cycle that traverses the different islands of player $i$'s domain by alternating between an incentive-compatibility constraint edge and a labelling constraint edge.
These cycles have length $2p$ and are given by
\begin{gather}
    \label{eq:bnom-guaranteed-ic-cycles}
    C_p^\beta = \{ \beta_{k,k} \rightarrow^{t_k} \beta_{k,k+1} \rightarrow^{t_k} \beta_{k+1,k+1} \, : \, k \in [p] \}, \\
    \label{eq:wnom-guaranteed-ic-cycles}
    C_p^\omega = \{ \omega_{k,k} \rightarrow^{t_k} \omega_{k,k+1} \rightarrow^{t_{k+1}} \omega_{k+1,k+1} \, : \, k \in [p] \},
\end{gather}
again with indices $-1$ and $p+1$ wrapping around appropriately.
The weights on the cycles in \eqref{eq:bnom-guaranteed-island-cycles} and \eqref{eq:wnom-guaranteed-island-cycles} are given by
$w(C_p^\beta) = \sum_{k \in [p]} t_k(f(\beta_{k,j})) - t_k(f(\beta_{k+1,j}))$
and 
$w(C_p^\omega) = \sum_{k \in [p]} t_k(f(\omega_{k-1,j})) - t_k(f(\omega_{k,j}))$,
respectively.
Since the weight of an edge $(x,y;t)$ is given by $t(f(x))-t(f(y))$ then the ``middle terms'' of the cycles in \eqref{eq:bnom-guaranteed-ic-cycles} and \eqref{eq:wnom-guaranteed-ic-cycles} cancel and their weight is equal to
$w(C_p^\beta) = \sum_{k \in [p]} t_k(f(\beta_{k,k})) - t_k(f(\beta_{k+1,k+1}))$
and
$w(C_p^\omega) = \sum_{k \in [p]} t_k(f(\omega_{k-1,k})) - t_k(f(\beta_{k,k+1}))$,
respectively.

There must exist a best-case labelling $\beta$ and worst-case labelling $\omega$ such that the resulting graphs $\mc{G}_{i,f}^\beta$ and $\mc{G}_{i,f}^\omega$ have no negative cycles in order for $f$ to be implementable as a NOM mechanism.
In the next section we will use the guaranteed cycles we have discussed in order to characterise the class of implementable functions for single-parameter agents.
\section{Single-Parameter Agents}

\label{sec:single-parameter}

In this section and the remainder of the paper we apply our analysis to single-parameter agents.
In this setting types can be described by a single number: when agent $i$ has type $t_i$ she values the outcome of $f$ on input $b$ as $t_i(f(b)) = t_i \cdot f_i(b)$, the product of her ``value per unit'' and her allocation, which allows us to decouple an agent's type from the outcome of the mechanism to prove properties of the latter in isolation.
This approach has seen much success in the literature: \cite{ArcherTardos2001} use it to show that designing truthful mechanisms reduces to designing monotonic algorithms (whether they are increasing or decreasing depends on whether agent types describe costs or valuations), and in such instances we can derive explicit formulae for truthful payments using the area under the allocation curve.
Other settings for single-parameter agents include combinatorial auctions \cite{Archer2003}, machine scheduling \cite{Andelman2006}, and payment computation \cite{Dobzinski2021}.

Our goal is to characterise the class of allocation functions that are implementable as NOM mechanisms for single-parameter agents.
In the case of strategyproofness it is well known that the allocation rule must be monotone in each player's bid for each fixed $b_{-i}$.
For our setting the ``shape'' of the allocation function will depend on the labelling, and we use the guaranteed cycles of the previous section to derive the following properties imposed by any labelling.
Take $p = 2$ in the guaranteed cycles $C_p^\beta$ and $C_p^\omega$ in \eqref{eq:bnom-guaranteed-island-cycles} and \eqref{eq:wnom-guaranteed-island-cycles}.
Then for any type $t_1, t_2 \in D_i$ the weights on the cycles in \eqref{eq:bnom-guaranteed-island-cycles} and \eqref{eq:wnom-guaranteed-island-cycles} we have
$(t_1 - t_2) (f_i(\beta_{1,1} - f_i(\beta_{2,2})) \ge 0$
and 
$(t_1 - t_2) (f_i(\omega_{2,1} - f_i(\omega_{1,2})) \ge 0$.
As for the the cycles in \eqref{eq:bnom-guaranteed-ic-cycles} and \eqref{eq:wnom-guaranteed-ic-cycles} we require
$(t_1 - t_2) (f_i(\beta_{1,j}) - f_i(\beta_{2,j}) \ge 0$
and
$(t_1 - t_2) (f_i(\beta_{2,j}) - f_i(\beta_{1,j}) \ge 0$, for any bid type $t_j \in D_i$.
For any labelling these inequalities must be satisfied for any two types from player $i$'s domain.
The following property describes when it is possible to find a valid labelling of $f$ that leads to a NOM mechanism for player $i$.

\begin{definition}[Overlapping]
    Let $O_{ij} = \{ f_i(t_j,b_{-i}) \, : \, b_{-i} \in D_{-i} \}$ denote the set of allocations of $f$ for player $i$ when she submits bid $t_j$, and let $O_i = \times_{j \in D_i} O_{ij}$.
    Allocation function $f$ is \emph{overlapping for $i$} if there exists $\mathbf{o}=(o_1,\ldots,o_d) \in O_i$ such that $o_1 \le o_2 \le \ldots \le o_d$, and  $f$ is \emph{overlapping} if it is overlapping for each player $i \in [n]$.
\end{definition}

In the proof of \Cref{thm:overlapping} we use the overlapping property to provide a labelling of $f$ that results in no negative cycles.
For this we use the pseudo-inverse function $\hat{f} : D_i \times O_{i} \to D$ which, given a bid $t_j$ of player $i$ and an allocation $o_j$ returns a bid profile $b = (t_j,b_{-i}) \in D$ such that $f_i(t_j,b_{-i}) = o_j$.

\begin{theorem}
    \label{thm:overlapping}
    Allocation function $f$ is either BNOM- or WNOM-implementable for single-parameter agents if and only if it is overlapping.
\end{theorem}

\begin{proof}
    ($\implies$)
    Assume for the sake of contradiction that there are two bids $t_1,t_2 \in D_i$ with $t_1 < t_2$ such that $f$ is not overlapping for $i$.
    That is, for any choice of $o_1 \in O_{i,1}$ we must have $o_2 < o_1$ for $o_2 \in O_{i,2}$.
    Since $f$ is WNOM-implementable then there exists a worst-case labelling $\omega$ such that the graph $\mc{G}_{i,f}^\omega$ contains no cycles of negative weight.
    By the monotone labelling property we require $(t_1-t_2)(f_i(t_1,\omega_{2,1})-f_i(t_2,\omega_{1,2})) \ge 0$, so $f_i(t_1,\omega_{2,1}) \le f_i(t_2,\omega_{1,2})$.
    Since $o_2 < o_1$ for any choice of $o_1$ then we must have $f_i(t_1,\omega_{j,1}) > f_i(t_2,\omega_{k,2})$ for all (not necessarily distinct) $t_j,t_k \in D_i$.
    Thus $f$ cannot be WNOM-implementable for $i$, a contradiction.
    Therefore $f$ must be overlapping for $i$.

    ($\impliedby$)
    We will use the overlaps of $f$ for player $i$ to define a valid worst-case labelling $\omega$ and note that the exact same reasoning carries over for best-case labellings $\beta$.
    Define iteratively the values $W(j) = \max \{ \min_{o_j \in O_{i,j}} \{ o_j \}, W(1), W(2), \ldots, W(j-1) \}$ for $j \in [d]$.
    Note that $W(j) \ge W(k) \iff t_j \ge t_k$.

    We will use $W$ to label $f$ with $\omega_{j,k} = \hat{f}_i(t_k,W(k))$ for each $t_j,t_k \in D_i$.
    Since for every bid $t_k$ we have $\omega_{j,k} = \omega_{\ell,k}$ for all types $t_k,t_\ell \in D_i$ we use a single subscript $\omega_k$ to denote this bid profile.
    We now show that no cycles that appear in the graph $\mc{G}_{i,f}^\omega$ are negative and hence $\omega$ is a valid worst-case labelling.
    
    First note that since for every bid $t_j$ we label the worst-case outcome for $i$ to occur at the same bid profile $\omega_j$ no matter her type, then each subgraph of $\mc{G}_{i,f}^\omega$ containing only nodes $(t_j,x_{-i})$ in which $i$ has bid $t_j$ will have no cycles, and these edges represent the labelling constraints. Hence we need only focus on edges which go from node $(t_j,\omega_j)$ to some $(t_k,\omega_k)$, which represent the incentive-compatibility constraints. Recall that the type annotation on each of these edges is always equal to the bid of player $i$ at the source node, so every edge has the form $(t_u,\omega_u) \rightarrow^{t_u} (t_v,\omega_v)$. Consider any cycle $C$ in $\mc{G}_{i,f}^\omega$ with weight $w(C) = \sum_{(u,v;t_u) \in C} t_u(f_i(\omega_u)-f_i(\omega_v))$. We will refer to negative weight edges in $C$ as ``uphill'' edges and positive weight edges as ``downhill'' edges. We can safely ignore edges of zero weight. Note that uphill edges occur when the output of $f$ for player $i$ is lower at the source than the destination profile, and since $W(j) \ge W(k) \iff t_j \ge t_k$ then all such edges go from a node in which $i$ bids a smaller type to one in which she bids a larger one. The opposite is true for downhill edges. For a given bid $t_j$ we will refer to the types $\max \{ t_k \in D_i \, : \, t_k < t_j \}$ and $\min \{ t_k \in D_i \, : \, t_k > t_j \}$ as the adjacent type to $t_j$, and adjacent nodes in $\mc{G}_{i,f}^\omega$ will be those in which $i$ submits adjacent types as bids to the mechanism. We will show that the weight of $C$ is bounded below by the weight of a (non-simple) cycle $C'$ which only contains edges between adjacent nodes, whose weight is in turn always non-negative.
    
    In the following we will use simply the label $\omega_j$ as shorthand for the bid profile $(t_j,\omega_j)$. First consider a downhill edge $e^+ = \omega_j \rightarrow^{t_j} \omega_{j-k}$ going directly from $\omega_j$ to $\omega_{j-k}$ that ``skips'' $k-1$ adjacent nodes. Since $t_j > t_{j-k}$ then $W(j) = f_i(t_j,\omega_j) > f_i(t_{j-k},\omega_{j-k}) = W(j-k)$. The weight $w(e^+)$ of such an edge is therefore:
    \begin{align*}
        w(e^+) & = t_j [ f_i(\omega_j) - f_i(\omega_{j-k}) ] \\
        & = t_j [ f_i(\omega_j) - f_i(\omega_{j-1}) + f_i(\omega_{j-1}) - f_i(\omega_{j-2}) + \ldots + f_i(\omega_{j-k+1}) - f_i(\omega_{j-k}) ] \\
        & \ge t_j(f_i(\omega_j)-f_i(\omega_{j-1})) + t_{j-1}(f_i(\omega_{j-1})-f_i(\omega_{j-2})) + \ldots + t_{j-k+1}(f_i(\omega_{j-k+1})-f_i(\omega_{j-k}))
    \end{align*}
    Following the same reasoning we can show that the weight of an uphill edge $e^- = (t_j,\omega_j) \rightarrow^{t_j} (t_{j+k},\omega_{j+k})$, in which we skip $k-1$ adjacent nodes, can be similarly lower bounded:
    \begin{align*}
        w(e^-) & = t_j [ f_i(\omega_j) - f_i(\omega_{j+k}) ] \\
        & = t_j [ f_i(\omega_j) - f_i(\omega_{j+1}) + f_i(\omega_{j+1}) - f_i(\omega_{j+2}) + \ldots + f_i(\omega_{j+k-1}) - f_i(\omega_{j+k}) ] \\
        & \ge t_j(f_i(\omega_j)-f_i(\omega_{j+1})) + t_{j+1}(f_i(\omega_{j+1})-f_i(\omega_{j+2})) + \ldots + t_{j+k-1}(f_i(\omega_{j+k-1})-f_i(\omega_{j+k}))
    \end{align*}
    Notice in both cases that the right hand side of the final inequality is the expression for the weight of a path between two non-adjacent nodes that only traverses edges between adjacent nodes.
    Now if we take any cycle $C$ and replace all such non-adjacent edges with a path comprising only adjacent edges then we will be left with a non-simple cycle that is composed of many two-cycles (see \Cref{fig:adjacent-type-cycles}).
    We know from the labelling monotonicity property of our labelling $\omega$ that all two-cycles are non-negative, hence so is their sum.
    This non-negative value is a lower bound on $w(C)$, thus there must be no negative cycles induced by labelling $\omega$ in $\mc{G}_{i,f}^\omega$ and by \Cref{thm:nom-implementable} then $f$ is implementable and the proof is complete.
\end{proof}

\begin{figure}[ht]
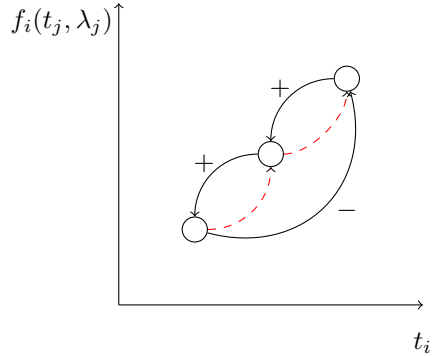

    \ctikzfig{tikzfigs/uphill-edges-graph}
    \caption{%
    The labelling given in the proof of \Cref{thm:overlapping} is always possible whenever $f$ is overlapping and will end up with edges that ``skip'' adjacent types.
    The weight of each cycle that uses such edges is no less than the weight of any cycle with edges only going between adjacent types, which in turn all have non-negative weight.
    This is true when replacing both ``uphill'' edges (pictured) and ``downhill'' edges between non-adjacent types with adjacent ones.
    Therefore the original labelling yields a valid NOM mechanism.
        }
    \label{fig:adjacent-type-cycles}
\end{figure}


%


As a corollary of \Cref{thm:overlapping} we can show that as long as $f$ is overlapping then it is straightforward to achieve both BNOM and WNOM simultaneously.
The only extra condition imposed is that unless the function $f$ is constant for player $i$ on bid $t_j$ -- that is, it allocates the same amount $f_i(t_j,\cdot)$ to $i$ whenever she bids $t_j$, no matter the bid profile of the other players -- then the best- and worst-case labelled profiles must be different.

We use the fact that $f$ is overlapping to prove sufficiency in \Cref{thm:overlapping} by defining a particular ``simple'' labelling in which the best- and worst-case outcomes for a given bid of player $i$ all occur at the same bid profile, no matter her type.
Such labellings, which we call \emph{single-line labellings}, allow us to precisely characterise the payments that lead to incentive-compatibility.

\begin{definition}[Single-line mechanism]
    A labelling $\lambda$ is \emph{single-line} if for every player $i$, for every bid $t_j \in D_i$ we have $\lambda_{kj} = \lambda_{\ell j}$ for every pair of types $t_k, t_\ell \in D_i$.
    A \emph{single-line mechanism} is one whose payments are defined on a social choice function $f$ whose labelling is single-line.
\end{definition}

When a labelling $\lambda$ is single-line then we will index it simply by the agent's bid: if player $i$ bids $t_j$ in a single-line mechanism then labelled profile associated with this bid is denoted $\lambda_j$.
After using these labellings to derive explicit payment formulae for single-line mechanisms we will then show in \Cref{sec:bilateral-trade} that more general labellings can lead to different payments.


\cite{ArcherTardos2001} show that a work function that allocates chores (where types describe \emph{costs} as opposed to valuations) admits truthful payments if and only if it is non-increasing and characterise these payments using the area under the allocation function.
\cite{Apt2022} recently provided an elementary proof of the uniqueness of these payment functions using a little-known result from the theory of real functions.
It is straightforward to derive a similar result in our setting for single-line mechanisms: we show that NOM payments exist if and only if one can find a non-decreasing collection of outcomes for the best- and worst-cases of the function.
Fixing such a labelling $\lambda$, we may treat $f$ as a function only in player $i$'s bid $b_i$ and derive the following payments by solving the resulting differential equations from \cite{ArcherTardos2001} in the same way.

\begin{theorem}
    \label{thm:explicit-payments}
    Social choice function $f$ is NOM-implementable if and only if $f$ is overlapping. When this is the case and the resulting mechanism $(f,p)$ uses a single-line labelling, the payments $p$ take the form
    \begin{gather}
        \label{eq:nom-explicit-payments}
        p_i(t_j,\lambda_j) = h_i(\lambda) - t_j f_i(t_j,\lambda_j) + \int_0^{t_j} f_i(u,\lambda_u) \, \mathrm{d} u,
    \end{gather}
    for each agent $i$, where $h_i(\lambda)=h_i(\lambda_1,\ldots,\lambda_d)$ is an arbitrary function that depends only on the (best- or worst-case) labelling $\lambda$.
\end{theorem}

Requiring IR and NPT further restricts the above payments.
For agents with valuation functions over $O$ we have $h_i(\beta) = h_i(\omega) = 0$, while for agents with cost functions we have $h_i(\beta) \ge \int_0^{t_j} f_i(u,\beta_u) \, \mathrm{d} u$ and $h_i(\omega) \ge \int_0^{t_j} f_i(u,\omega_u) \, \mathrm{d} u$ for all $b \in D_i$.
The remaining payments (i.e., for bid profiles which are not best- or worst-case inputs for player $i$) must simply satisfy the labelling constraints induced by $\beta$ and $\omega$.
For each $t_j,t_k \in D_i$ we simply require $t_j(f_i(\beta_j)) + p_i(\beta_k) \ge t_j(f_i(t_k,b_{-i})) + p_i(t_k,b_{-i})) \ge t_j(f_i(\omega_k)) + p_i(\omega_k)$ for each $b_{-i} \in D_{-i}$.
It is important to note the transition from finite to infinite domains, done both to mirror the setting of \cite{ArcherTardos2001} and to use in the upcoming section.
Conceptually this is not an issue and in the case of finite domains we can simply replace the above integral with a summation.

\section{Bilateral Trade}

\label{sec:bilateral-trade}

We now apply our previous analysis to bilateral trade.
In this setting there are two agents, $B$ and $S$, where $B$ is a potential buyer of an item that $S$ may produce and sell.
The buyer has a valuation $v \in D_B$ for the item and the seller incurs a cost $c \in D_S$ for producing the item.
We assume $D_B = D_S = [0,1]$, and let $D = D_B \times D_S$.
A mechanism for bilateral trade is a tuple $M=(f,p)$ where $f : D \to \{0,1\}$ indicates whether a trade takes place and $p : D \to \mathbb{R}_{\ge 0}$ describes the payments.
The buyer's utility is therefore given by $u_B(M(x,y)) = v \cdot f(x,y) - p_B(x,y)$ and the seller's by $u_S(M(x,y)) = p_S(x,y) - c \cdot f(x,y)$.

In addition to individual rationality (defined in \Cref{sec:preliminaries}) we focus on the following properties: a mechanism $M = (f,p)$ is \emph{efficient} when $f(x,y) = 1$ if and only if $x \ge y$ for all $(x,y) \in D$; and \emph{weakly budget balanced} if $p_S(x,y) \le p_B(x,y)$ for all $(x,y) \in D$.
\cite{Troyan2020} show that every efficient, individually rational, weakly budget balanced mechanism for bilateral trade is obviously manipulable, hence we relax budget balance and say that $M$ is \emph{$\alpha$-weakly budget balanced} if $p_S(x,y) \le \alpha \cdot p_B(x,y)$ for $\alpha \ge 1$.
We begin by showing that single-line mechanisms cannot resolve the impossibility result with any bounded subsidy.

\subsection{Single-line mechanisms require unbounded subsidies}

We first note the form taken by the explicit payment formulae for single-line mechanisms from \Cref{thm:explicit-payments}.
In this setting the buyer must pay $p_B(x,\lambda^B_x) = h_B(\lambda^B) + x f(x,\lambda^B_x) - \int_0^x f(u,\lambda^B_u) \, \mathrm{d} u$ while the seller must receive $p_S(\lambda^S_y,y) = h_S(\lambda^S) + y f(\lambda^S_y,y) - \int_0^y f(\lambda^S_u,u) \, \mathrm{d} u$, where $\lambda^B$ is any labelling for the buyer (and thus) represents a bid of the seller and $\lambda^S$ is any labelling for the seller (and thus represents a bid of the buyer).
Again notice that if we have IR and NPT then $h_B(\lambda^B)$ goes to zero, while for the seller $h_S(\lambda^S)$ must be at least the area under the curve.

The following facts will be useful.
Let $x$ and $y$ be bids of the buyer and seller, respectively.
For any efficient IR mechanism we require  $p_B(x,y) \le x$ and $p_S(x,y) \ge y$ for all $x \ge y$, while for $x < y$ we have $p_B(x,y) = 0$.
Therefore whenever trade does not take place the utility of the buyer is 0 no matter her valuation.
In \Cref{lem:bnom-wnom-threshold} we show that for a single-line mechanism the curves representing the best- and worst-case outcomes of $f$ are constant.
Since the buyer values the item while the seller incurs a cost for it then by \Cref{thm:overlapping} the best- and worst-case outcomes of $f$ must be non-decreasing for the buyer and non-increasing for the seller.
If $M=(f,p)$ is NOM and $f$ takes values in $\{0,1\}$ then there must be a threshold at which, for all bids that are at least this threshold for some player, their best- or worst-case allocation from the mechanism flips from one value to the other.
We formalise this in the following lemma and show for both the buyer and the seller that these thresholds must be placed at the extremes, namely, 0 or 1.

\begin{lemma}
    \label{lem:bnom-wnom-threshold}
    In any efficient, IR, single-line BNOM mechanism the best-case outcomes always occur for both the buyer and the seller when the trade is executed, while in any efficient, IR, single-line WNOM mechanism the worst-case outcomes always occur for both the buyer and the seller when the trade is not executed.
\end{lemma}

\begin{proof}
    Let $M=(f,p)$ be an efficient, IR, single-line NOM mechanism and first consider the threshold for best-case outcomes from the buyer's side.
    Since $f$ is non-decreasing let $\theta_B$ be a threshold bid of the buyer such that $f(x,\beta^B_x) = 1$ for all $x \ge \theta_B$.
    Now take $x < \theta_B$ and let $y$ be a bid of the seller such that $f(x,y)=1$.
    By the BNOM labelling constrains we have $u_B(M(x,\beta^B_x)) = 0 \ge v - p_B(x,y)$ for all $v \in D_B$, so $1 \le p_B(x,y) \le x$, where the second inequality follows from IR and efficiency.
    Now we have $1 \le x < \theta_B$ where $\theta_B \in [0,1]$, a contradiction.
    Hence there is no $x$ such that $x < \theta_B$, meaning $\theta_B = 0$ and trade always occurs in the buyer's best-case outcome.
    
    Now consider the seller.
    Since $f$ must be non-increasing in her bid there again must be a threshold $\theta_S \in [0,1]$ such that $f(\beta^S_y,y)=1$ for all $y \le \theta_S$ and 0 otherwise.
    For some bid $y > \theta_S$ of the seller her best case utility is $u_S(M(\beta^S_y,y)) = p_S(\beta^S_y,y) = h_S(\beta^S) + y f(\beta^S_y,y) - \int_0^y f(\beta^S_u,u) \, \mathrm{d} u$.
    Since $c$ describes the seller's cost then $h_S(\beta^S)$ must be at least the area under the curve.
    Setting $h_S(\beta^S)$ to this area gives $p_S(\beta^S_y,y) = y f(\beta^S_y,y) + \int_y^1 f(\beta^S_u,u) \, \mathrm{d} u$ so her best-case utility is 0.\footnote{%
    We can of course pay the seller more, as long as it doesn't interfere with incentives.
    We will be using this lemma to prove the impossibility of (bounded) $\alpha$-WBB, so if we can show this when we are being frugal then the claim will hold even when using more generous payments.
    Thus we set $h_S(\beta)$ as small as possible -- exactly the area under the curve -- and this is enough to prove the unboundedness.%
    }
    Now take any $x$ such that $f(x,y)=1$.
    By the BNOM constraints we have $u_S(M(x,y)) = p_S(x,y) - c \le 0$ for all $c$.
    Combining this with individual rationality we have $c \ge p_S(x,y) \ge y$, which gives us a contradiction if we take $c < y$.
    Thus $y \le \theta_S$ for all $y$ and so $\theta_S = 1$, meaning trade always occurs in the seller's best-case outcome.
    
    We can apply similar reasoning to show that trade does not occur for the buyer or seller in their worst-case outcome, no matter their bid.
    Since the mechanism is NOM then $f$ must be overlapping, hence there is a threshold bid $\theta_B$ and $\theta_S$ for the buyer and seller respectively whereby for all $x \ge \theta_B$ we have $f(x,\omega^B_x) = 1$ and $0$ otherwise, and for all $y \le \theta_S$ we have $f(\omega^S_y,y) = 1$ and $0$ otherwise.
    First consider the buyer.
    Assume there is a bid $x > \theta_B$ of the buyer and $y$ a bid of the seller such that $f(x,y)=0$.
    The labelling constraints require that $v - p_B(x,\omega^B_x) \le 0$ for all $v \in D_B$, and combining this with the previous facts we get $p_B(x,\omega^B_x) = x$.
    Now taking $v > x$ the buyer's utility at her worst-case outcome is strictly positive, a contradiction.
    Hence $\theta_B \ge x$ for all $x$ and thus $\theta_B = 1$.
    Therefore in the buyer's worst case the trade is never executed.
    
    Finally now focus on the worst cases of the seller, let $y < \theta_S$ and let $x$ be a bid of the buyer such that $f(x,y) = 0$.
    Then $p_S(x,y) = 0$, and since trade occurs in the seller's worst case by assumption then $p_S(\omega^S_y,y) \le y$.
    By \eqref{eq:nom-explicit-payments} then $p_S(\omega^S_y,y) \ge h_S(\omega^S) + y f(\omega^S_y,y) - \int_0^y f(\omega^S_u,y) \, \dx{u} = h_S(\omega^S) + y - y \ge \theta_S$, where the final inequality follows from the fact that $h_S(\omega^S)$ must be at least the area under the curve for any IR and NPT mechanism.
    Thus $p(\omega^S_y,y) \ge \theta_S$ and $p_S(\omega^S_y,y) \le y$ where $y < \theta_S$, a contradiction.
    So $\theta = 1$ and the seller's worst case always occurs when the trade is not executed.
\end{proof}

\begin{figure}[h]
    \centering
    \def\twidth{0.45}
    \subfloat[Threshold bid $\theta_B$ for the buyer. For best cases we have $\theta_B = 0$ and for worst cases we have $\theta_B = 1$.]{%
        \label{fig:bilateral-trade-buyer-threshold}
        \begin{tikzpicture}
	\begin{pgfonlayer}{nodelayer}
		\node [style=none] (0) at (-4, 3) {};
		\node [style=none] (1) at (-4, -4) {};
		\node [style=none] (2) at (4, -4) {};
		\node [style=none] (3) at (-2, 1) {};
		\node [style=none] (4) at (3, 1) {};
		\node [style=none] (5) at (-2, -4) {};
		\node [style=none] (6) at (-2, -5) {$\theta_B$};
		\node [style=none] (7) at (4.5, -4.5) {$x$};
		\node [style=none] (8) at (-5.5, 3.25) {$f(x,\lambda^B_x)$};
		\node [style=none] (9) at (-5, 1) {$1$};
		\node [style=none] (10) at (-5, -4) {$0$};
		\node [style=none] (11) at (3, -5) {$1$};
		\node [style=none] (12) at (3, -3.75) {};
		\node [style=none] (13) at (3, -4.25) {};
	\end{pgfonlayer}
	\begin{pgfonlayer}{edgelayer}
		\draw [style=Directed] (1.center) to (2.center);
		\draw (4.center) to (3.center);
		\draw [style=Dashed] (3.center) to (5.center);
		\draw [style=Directed] (1.center) to (0.center);
		\draw (12.center) to (13.center);
	\end{pgfonlayer}
\end{tikzpicture}
    } \hfil
    \subfloat[Threshold bid $\theta_S$ for the seller. For best cases we have $\theta_S = 1$ and for worst cases we have $\theta_S = 0$.]{%
        \label{fig:bilateral-trade-seller-threshold}
        \begin{tikzpicture}
	\begin{pgfonlayer}{nodelayer}
		\node [style=none] (0) at (-4, 3) {};
		\node [style=none] (1) at (-4, -4) {};
		\node [style=none] (2) at (4, -4) {};
		\node [style=none] (3) at (1, 1) {};
		\node [style=none] (4) at (-4, 1) {};
		\node [style=none] (5) at (1, -4) {};
		\node [style=none] (6) at (1, -5) {$\theta_S$};
		\node [style=none] (7) at (4.5, -4.5) {$y$};
		\node [style=none] (8) at (-5.5, 3.25) {$f(\lambda^S_y,y)$};
		\node [style=none] (9) at (-5, 1) {$1$};
		\node [style=none] (10) at (-5, -4) {$0$};
		\node [style=none] (11) at (3, -5) {$1$};
		\node [style=none] (12) at (3, -3.75) {};
		\node [style=none] (13) at (3, -4.25) {};
	\end{pgfonlayer}
	\begin{pgfonlayer}{edgelayer}
		\draw [style=Directed] (1.center) to (2.center);
		\draw (4.center) to (3.center);
		\draw [style=Dashed] (3.center) to (5.center);
		\draw [style=Directed] (1.center) to (0.center);
		\draw (12.center) to (13.center);
	\end{pgfonlayer}
\end{tikzpicture}}
    
    \caption{Thresholds bids used in \Cref{lem:bnom-wnom-threshold}.}
    \label{fig:bilateral-trade-threshold-bids}
\end{figure}
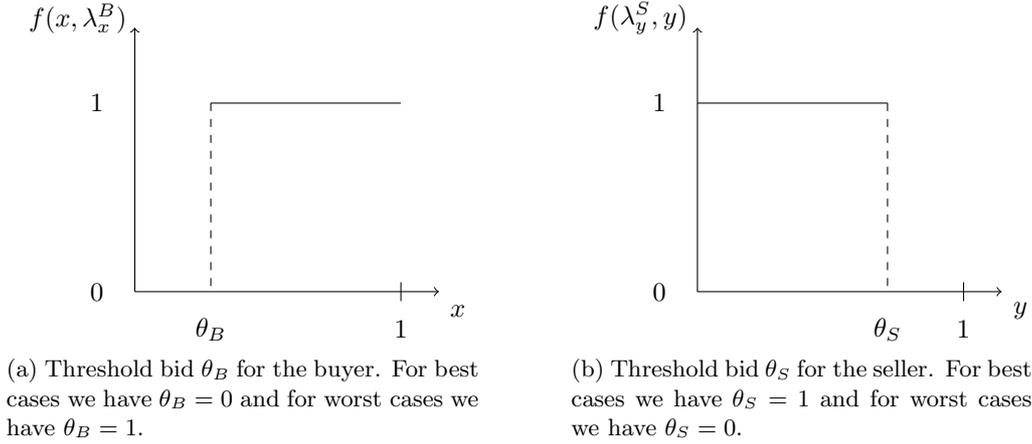

As a corollary of \Cref{lem:bnom-wnom-threshold} we have $p_B(x,\omega^B_x) = p_S(\omega^S_y,y) = 0$ for all $x,y$ since trade does not occur at the worst-case profiles.
Conversely the buyer's best outcome must be when she gets the item and pays nothing, while the seller's must occur when she produces the item and is paid at least her bid.
For any IR, NPT, single-line BNOM mechanism $M$ best-case payments are as in \eqref{eq:nom-explicit-payments}, so $p_B(x,\beta^B_x) = 0$ and $p_S(\beta^S_y,y) \ge y + \int_y^1 f(\beta^S_u,u) \, \mathrm{d} u = 1$.
The following result shows that single-line mechanisms are not flexible enough to avoid the impossibility of \cite{Troyan2020} without unbounded subsidy.

\begin{theorem}
    \label{thm:bnom-wnom-single-line-unbounded-alpha}
    Any efficient, IR, single-line BNOM or WNOM mechanism with $\alpha$-WBB has unbounded $\alpha$.
\end{theorem}

\begin{proof}
    We first prove the result for BNOM mechanisms.
    From \Cref{lem:bnom-wnom-threshold} we know that $f(x,\beta^B_x) = f(\beta^S_y,y) = 1$ for any $x \in D_B$ and $y \in D_S$, while $p_B(x,\beta^B_x) = 0$ and $p_S(\beta^S_y,y) \ge y$.
    Since in the best case the buyer pays zero while the seller receives at least her cost, we may prove the claim by providing a bid profile that appears on the best-case curves of both the buyer and seller.
    Take any bid $x$ of the buyer and let $y = \beta^B_x > 0$ be a bid of the seller.
    Since $f(x,y)=1$ then $p_S(x,y) \ge y$ by efficiency and IR.
    The labelling constraints require $u_S(M(x,y)) \ge u_S(M(\beta^S_y,y))$ for all $y$ and so $p_S(\beta^S_y,y) \ge p_S(x,y) \ge y$.
    Since $y = \beta^B_x$ then $p_S(\beta^S_y,y) \ge p_S(x,\beta^B_x) \ge y > 0$, while $p_B(x,\beta^B_x) = 0$.
    Therefore at $(x,\beta^B_x)$ the buyer pays 0 while the seller must receive some positive amount.

    
    Now let $M=(f,p)$ be an efficient, IR, single-line WNOM mechanism and consider a bid $x \in D_B$ of the buyer for which $f(x,\omega^B_x) = 0$ and thus $p_B(x,\omega^B_x) = u_B(M(x,\omega^B_x)) = 0$.
    Let $y \in D_S$ be a bid of the seller such that $f(x,y)=1$ and note that by efficiency $x \ge y$.
    Since the worst utility of the buyer is 0 then $p_B(x,y) \le v$ and thus we need $p_B(x,y) \le \min \{ x \, : \, x \ge y \} = y$.
    Now consider the seller with bid $y \in D_S$ such that $f(\omega^S_y,y) = 0$.
    Let $x \in D_B$ be a bid of the buyer such that $f(x,y) = 1$.
    By the labelling constraints we have $p_S(x,y) - c \ge p_S(\omega^S_y,y) \ge 0$, hence $p_S(x,y) \ge c$ for all $c$ (since $M$ is single-line) and thus $p_S(x,y) \ge \max \{ y \, : \, x \ge y \} = y$.
    
    Now consider the bid profile $(1,\varepsilon)$ for arbitrarily small $\varepsilon > 0$.
    The buyer must pay at most $\varepsilon$ while the seller is to receive at least 1, giving the desired result.
\end{proof}

\subsection{Avoiding subsidies by changing the labelling}

Single-line mechanisms are overly restrictive if we want to design efficient, IR, and $\alpha$-WBB NOM mechanisms for bilateral trade, so we now explore if the issue of unbounded subsidy is avoided by considering more general labellings and hence payments.
We show that while this works in a strong sense for WNOM mechanisms, where we can achieve exact (i.e., $\alpha = 1$) budget balance, surprisingly unbounded subsidies are inherent in any BNOM mechanism.
For the following two results it is enough to consider domains $D_B = [0,1]$ and $D_S = \{ 0,1 \}$.

\begin{theorem}
    \label{thm:wnom-bounded-subsidy}
    There is an efficient, IR, WBB, WNOM mechanism.
\end{theorem}

\begin{proof}
    Consider the labelling $\omega^B$ for the buyer such that $\omega^B_{jk} = 1$ if and only if $t_j \ge t_k$, and 0 otherwise.
    Observe that for all $x \in D_B$ every cycle within the $x$-island of $\mc{G}_{B,f}^\omega$ have non-negative weight: the $x$-island contains exactly two nodes $(x,0)$ and $(x,1)$, and $\omega^B$ ensures that each two-cycle has weight at least $x (f(x,0) - f(x,1)) + (x - \varepsilon) (f(x,1) - f(x,0)) = \varepsilon > 0$ (i.e., for $t_j = x$ and $t_k = x - \varepsilon$).
    Recall that incentive-compatibility edges are the only ones to traverse different islands, and note with the following that incorporating these into a cycle of labelling edges results in no negative cycles.
    The shortest such cycle has the form $(x,1) \rightarrow^{x} (x',0) \rightarrow^{x'} (x',1) \rightarrow^{x'} (x,1)$ for $x' > x$, and its weight is minimised when $x' = x+\varepsilon$ for some $\varepsilon > 0$; the weight of the cycle therefore always non-negative.
    It is straightforward to verify that these are the only types of cycle induced by the labelling, and that increasing the length of the cycle will only increase its weight.
    Thus $f$ is WNOM-implementable.
    It now remains to derive the payments this labelling produces.
    
    Fix a bid $t_j$ of the buyer.
    Since $\omega_{jj}=1$ then trade does not occur and we have $u_B(M(t_j,1))=0$ for all $v$.
    Since $\omega_{ij}=0$ for $t_i < t_j$ then trade occurs and we have $v - p_B(t_j,0) \le u_B(M(t_j,1)) = 0$.
    Thus for all $v < t_j$ we have $p_B(t_j,0) \ge v$ and so $p_B(t_j,0)$ tends to $t_j$.
    Combining this with IR we have $t_j \le p_B(t_j,0) \le t_j$ and therefore when the trade occurs the buyer pays her bid, and 0 otherwise.
    Therefore this labelling induces first-price payments for the buyer.
    Now consider payments for the seller where she gets paid her own bid when trade occurs.
    It is straightforward to verify that these payments for the seller are strategyproof and therefore NOM by \Cref{thm:SP-implies-NOM}.
    When trade occurs the payments are $p_B(x,y) = x$ and $p_S(x,y) = y$, and 0 otherwise.
    Since trade occurs only when $x \ge y$ then the payment given to the seller never exceeds the amount extracted from the buyer, hence we achieve (exact) weak budget balance.    
\end{proof}
    
With the following we show that moving away from single-line mechanisms does not avoid the issue of unbounded subsidies for BNOM as it does for WNOM.
The proof, which we omit in the interest of space, is by case analysis on all possible labellings between two types of the buyer's domain and reveals that no labelling results in a mechanism that bypasses the negative result, even if we restrict to these two types.
Each labelling falls into one of four cases: it is single-line, for which we have already proved subsidies are unbounded; it is invalid and leads to a negative cycle, meaning NOM payments are impossible by \Cref{thm:nom-implementable}; it induces no negative cycles but otherwise leads to contradiction when taken with IR and efficiency; or it leaves us with the same payments as in the single-line case, meaning unbounded subsidies are unavoidable.

\begin{theorem}
    \label{thm:bnom-unbounded-subsidy}
    There is no efficient, IR, $\alpha$-WBB, BNOM mechanism for any $\alpha < \infty$.
\end{theorem}

\begin{proof}
    Let $f$ be an efficient allocation function for bilateral trade. We will prove the claim by case analysis of all possible labellings on the subgraph of $G_{B,f}^\beta$ containing only a $j$-island and a $k$-island for some types $t_j,t_k \in D_B$. Each general labelling falls into one of four cases: it is single-line, for which we have already proved subsidies are unbounded; it is invalid and leads to a negative cycle $G_{B,f}^\beta$, meaning payments ensuring BNOM are impossible; it induces no negative cycles in $G_{B,f}^\beta$ but otherwise leads to contradiction when taken with IR and efficiency; or it leaves us with the same payments as in the single-line case, meaning unbounded subsidies are unavoidable. The first two cases are straightforward so we will focus on the latter two.
    
    Let $D_B=[0,1]$ and $D_S=\{0,1\}$. First note that in any valid labelling $\beta$ each $j$-island of $G_{B,f}^\beta$ will consist only of two nodes, $(t_j,0)$ and $(t_j,1)$, hence for any labelling to be valid it must be that $f(t_j,\beta_{ij}) \ge f(t_j,\beta_{kj}) \iff t_i \ge t_k$. Now fix $t_j,t_k \in D_i$ with $t_j < t_k$ and observe that by IR we have $p_B(t_j,0) \le t_j$, so either $\beta_{jj} = (t_j,0)$ or $\beta_{jj} = (t_j,1)$ and $p_B(t_j,0) = t_j$. The two cases for $\beta$ we will discuss are as follows.
        
    \begin{figure}[ht]
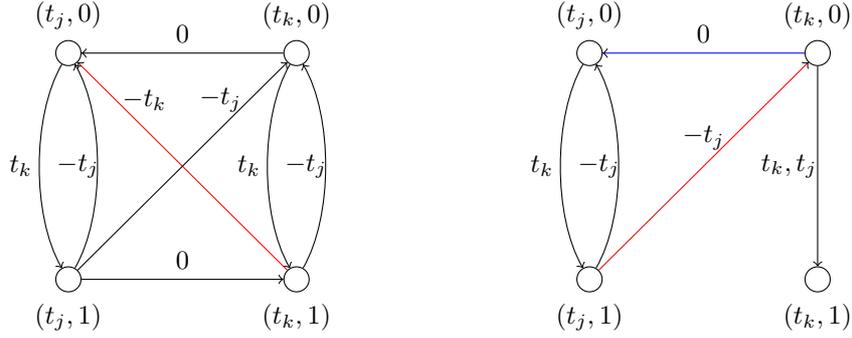

        \centering
    	\begin{subfigure}{.4\textwidth}
            \centering
            \ctikzfig{tikzfigs/bad-beta-1}
            \caption{A valid BNOM labelling incompatible with IR and efficiency.}
            \label{fig:bad-beta-1}
    	\end{subfigure} \quad
    	\begin{subfigure}{.4\textwidth}
            \centering
            \ctikzfig{tikzfigs/bad-beta-2}
            \caption{A valid BNOM labelling that yields single-line payments.}
            \label{fig:bad-beta-2}
    	\end{subfigure}
    	\caption{No BNOM labelling can achieve bounded subsidies in an efficient, IR mechanism for bilateral trade.}
    	\label{fig:bad-beta}
    \end{figure}
    
    \emph{Case 1: $\beta$ is incompatible with IR and efficiency.} Consider $\beta$ where $\beta_{jj} = (t_j,1)$, $\beta_{kj} = (t_j,0)$, $\beta_{kk} = (t_k,1)$, and $\beta_{jk} = (t_k,0)$. By the incentive-compatibility constraints (represented by the red edge in \Cref{fig:bad-beta-1}) we have $0 \ge t_k - p_B(t_j,0) = t_k - t_j$, a contradiction since $t_j < t_k$. 
    
    \emph{Case 2: $\beta$ leads to single-line payments.} Consider $\beta$ where $\beta_{jj} = \beta_{kj} = (t_j,0)$, $\beta_{kk} = (t_k,0)$, and $\beta_{jk} = (t_k,1)$ and whose graph $G_{B,f}^\beta$ is given in \Cref{fig:bad-beta-1} in which edges are annotated with their weights. By incentive-compatibility constraints we have $0 \ge t_j - p_B(t_k,0)$ (represented by red edge in \Cref{fig:bad-beta-2}) and hence $p_B(t_k,0) \ge t_j$, as well as $t_k - p_B(t_k,0) \ge t_k - p_B(t_j,0)$ (represented by the blue edge in \Cref{fig:bad-beta-2}) and so $p_B(t_k,0) \le p_B(t_j,0)$. Recall that we have $p_B(t_j,0) = t_j$ and therefore $t_j \le p_B(t_k,0) \le p_B(t_j,0) = t_j$, so $p_B(t_k,0) = t_j$. Now take $t_j = 0$ and we see that payments in a profile in which trade occurs are identically zero. This is precisely the issue that leads to unbounded subsidies.
    
    Each possible labelling, even when only considering two types $t_j$ and $t_k$, falls into one of the four above mentioned cases. Thus the only valid labellings that are not incompatible with IR and efficiency lead to exactly the same payments as we get with a single-line mechanism. Hence subsidies remain unbounded for arbitrary labellings for any efficient, IR, BNOM mechanism.
\end{proof}

\subsection{Characterising IR, WBB, NOM bilateral trade mechanisms}

We close our analysis of NOM bilateral trade mechanisms by providing a full characterisation of the class of IR, WBB, NOM mechanisms for this setting.
We introduce some notation to conveniently formulate our characterisation.
For a number $x$ we denote by $\text{succ}_B(x)$ and $\text{prec}_B(x)$ respectively the least element in $D_B$ strictly greater than $x$ and the greatest element strictly less than $x$.
We use $\text{succ}_S(y)$ and $\text{prec}_S(y)$ analogously for the seller's domain $D_S$.
Let $H_B$ denote the least element in $D_B$ that is greater than every element in $D_S$, and $L_B$ the greatest element in $D_B$ that is less than every element in $D_S$.
For a mechanism we may partition the type space into three sets: $M_0^B$, $M_1^B$, and $M_{01}^B$.
$M_0^B$ consists of all bids of the buyer for which trade is guaranteed not to happen, $M_1^B$ consists of all bids for which trade surely happens, and $M_{01}^B$ consists of the remaining elements of $D_B$, where trade depends on the seller's bid.
The sets $M_0^S$, $M_1^S$, and $M_{01}^S$ can be defined analogously for the seller.

For a mechanism $M$ define the \emph{utility interval} $I_B(M,v,x)$ for the buyer as the two numbers $(\ell,h)$ such that $\ell = \min \{ u_B(M(x,y)) \, : \, y \in D_S \}$ and $h = \max \{ u_B(M(x,y)) \, : \, y \in D_S \}$, and define $I_S(M,c,y)$ analogously for the seller.
The NOM property imposes that $I_B(M,v,v) \ge I_B(M,v,x)$ for all $v,x \in D_B$ and that $I_S(M,c,c) \ge I_S(M,c,y)$ for all $c,y \in D_S$.

\begin{theorem}
    \label{thm:trading-windows}
    A deterministic bilateral trade mechanism $M$ for domains $D_B,D_S$ is IR, WBB, and NOM if and only if there are thresholds $p_B^{\min},p_B^{\max},p_S^{\min},p_S^{\max}$, where $p_B^{\min} \le p_B^{\max}$ and $p_S^{\min} \le p_S^{\max}$, such that:
    \begin{enumerate}
        \item When trade doesn't occur the buyer's and the seller's price are both zero. When trade occurs the buyer's price always exceeds the seller's price, both prices are non-negative, the buyer's price is less than their valuation, and the seller's price is at least the buyer's valuation.
        
        \item The set $M_0^B$ contains all types in $D_B$ less than $p_B^{\min}$, the set $M_1^B$ contains all types in $D_B$ greater than $p_B^{\max}$, and $M_{01}^B$ contains all types in $D_B$ in between $p_B^{\min}$ and $p_B^{\max}$. If $p_B^{\min}$ is itself in $D_B$ then $p_B^{\min}$ is in either $M_0^B$ or $M_{01}^B$. Similarly if $p_B^{\max}$ is in $D_B$ then $p_B^{\max}$ is in either $M_1^B$ or $M_{01}^B$.
        
        \item For every type $x \in M_{01}^B \cup M_1^B$ the buyer's payment in outcome $M(x,y)$ is at least $p_B^{\min}$ for all $y \in D_S$, and there exists some $y \in D_S$ such that in outcome $M(x,y)$ trade occurs and the buyer pays $p_B^{\min}$. For every type $x \in M_1^B$ the buyer's payment in outcome $M(x,y)$ is at most $p_B^{\max}$ for all $y \in D_S$, and there exists a $y \in D_S$ such that in outcome $M(x,y)$ trade occurs at price $p_B^{\max}$.
        
        \item (Analogous to point 2 but for the seller) The set $M_0^S$ includes all types in $D_S$ greater than $p_S^{\max}$, the set $M_1^S$ includes all types in $D_S$ less than $p_S^{\min}$, and the set $M_{01}^S$ consists of all types in $D_S$ in between $p_S^{\min}$ and $p_S^{\max}$. If $p_S^{\min}$ itself is in $D_S$, then $p_S^{\min}$ is either in $M_1^S$ or $M_{01}^S$. Similarly, if $p_B^{\max}$ is in $D_S$ then $p_B^{\max}$ is in $M_{01}^S$ or $M_1^S$.
        
        \item (Analogous to point 3 but for the buyer) For every type $y \in M_{01}^S \cup M_1^S$ there is some $x \in D_B$ such that in outcome $M(x,y)$ trade occurs with seller's price $p_S^{\max}$. For every type $y \in M_1^S$ there is furthermore a $x \in D_B$ such that in outcome $M(x,y)$ trade occurs at price $p_S^{\min}$.
    \end{enumerate}
\end{theorem}

\begin{proof}
    ($\implies$) Assume that $M$ is IR, WBB, and NOM. Point 1 follows directly from IR and WBB. We will prove points 2 and 3, and since the proofs of points 4 and 5 follow analogous reasoning they are omitted. We assume for simplicity of exposition that $M_0^B$, $M_{01}^B$, and $M_1^B$ are all non-empty. In case one or more of these sets are empty, the proof reduces to a simpler one. 
    
    Let $x_0^{\max}$ be the greatest type in $M_0^B$ and let $x$ be an arbitrary type less than $x_0^{\max}$. We show that $x$ is in $M_0^B$. Suppose the contrary and let $y$ be a seller's type such that trade occurs in outcome $M(x,y)$. The buyer's price in $M(x,y)$ is at most $x$ by IR, and hence it is less than $x_0^{\max}$. This means that the second element of $I_B(M,x_0^{\max},x)$ is positive, while $I_B(M,x_0^{\max},x_0^{\max}) = (0,0)$; a contradiction to NOM. We have now established that $x$ is in $M_0^B$ for all $x$ in $D_B$ where $x \le x_0^{\max}$.
    
    Let $x_1^{\min}$ be the least type in $M_1^B$ and let $x$ be an arbitrary type greater than $x_1^{\min}$. We show that $x$ is in $M_1^B$. Suppose the contrary; then the first element of $I_B(M,x,x)$ is 0. However note that any buyer's price charged when reporting $x_1^{\min}$ is less than $x_1^{\min}$ itself by IR, and hence is less than $x$. This would mean that the first element of $I_B(M,x,x_1^{\min})$ is positive; a contradiction to NOM. We have now established that $x$ is in $M_1^B$ for all $x$ in $D_B$ where $x \ge x_1^{\min}$.

    The above altogether shows that $M_0^B$, $M_{01}^B$, $M_1^B$ partition $D_B$ into three contiguous ranges: $M_0^B$ consists of all valuations in $D_B$ up to and including $p_0^{\max}$, $M_{01}^B$ consists of all valuations in $D_B$ strictly in between $p_0^{\max}$ and $p_1^{\min}$, and $M_1^B$ consists of all valuations in $D_B$ from $p_1^{\min}$ and up.
    
    Let $x$ be a type in $M_{01}^B \cup M_1^B$, and let $y$ be the seller's bid such that the buyer's price is minimised in outcome $M(x,y)$. Denote this price by $p(x)$.  Thus we have that the second element of $I_B(M,x,x)$ is $x - p(x)$. By NOM, for any other bid $x'$ we have that the second element of $I_B(M,x,x')$ is at most $x - p(x)$. In particular for $x'$ in $M_{01}^B \cup M_1^B$ it holds that the second element of $I_B(M,x',x')$ is $x' - p(x')$ and thus the second element of $I_B(M,x,x')$ is $x - p(x')$. Thus, $x - p(x') \le x - p(x)$, or equivalently $p(x) \le p(x')$. Since this inequality holds for all $x,x'$ in $M_{01}^B \cup M_1^B$, we have that $p(\cdot)$ is constant over $M_{01}^B \cup M_1^B$, and we define this constant as $p_B^{\min}$. Observe that $p_B^{\min}$ is at most the least type in $M_{01}^B$, by IR (i.e., the type $\text{succ}_B(x_0^{\max}))$. Furthermore $p_B^{\min} \ge x_0^{\max}$ because otherwise the second element of $I_B(M, x_0^{\max}, \text{succ}_B(x_0^{\max}))$ would equal $x_0^{\max} - p_B^{\min} \ge x_0^{\max} - x_0^{\max} = 0$, where 0 is the second element of $I_B(M,x_0^{\max},x_0^{\max})$; a contradiction to NOM. This establishes that the value $p_B^{\min}$ lies in the interval $[x_0^{\max}, \text{succ}_B(x_0^{\max})]$, as required from points 2 and 3 in the claim.
    
    Let $x$ now be a type in $M_1^B$, and let $y$ be the seller's bid such that the buyer's price is maximised in outcome $M(x,y)$. Denote this price by $P(x)$. Thus, the first element of $I_B(M,x,x)$ is $x - P(x)$, for all $x \in M_1^B$. Also, this implies that the first element of $I_B(M,x,x')$ is $x - P(x')$ for all $x,x' \in M_1^B$, which in turn yields that $P(x) \le P(x')$ for all $x,x' \in M_1^B$, and hence $P(\cdot)$ is constant on $M_1^B$. Let this constant be the value of $p_B^{\max}$.

    Note that $p_B^{\max}$ is at most $x_1^{\min}$, by IR. Furthermore, $p_B^{\max} \ge \text{prec}_B(x_1^{\min})$, as otherwise the first element of $I_B(M, \text{prec}_B(x_1^{\min}), x_1^{\min})$ would be positive. The first element of $I_B(M, \text{prec}_B(x_1^{\min}), \text{prec}_B(x_1^{\min}))$ is 0; contradicting with NOM. This proves that points 2 and 3 of the claim hold.
    
    ($\impliedby$) Assume that the 5 points of the claim hold for a mechanism $M$. We show that M is IR, WBB, and NOM. The first two properties, IR and WBB, follow trivially from point 1. We next show that NOM holds for the buyer. The proof for the NOM-property for the seller is analogous and is therefore omitted.

    Let $x$ be a bid in $M_0^B$ so that $I_B(M,x,x) = (0,0)$. Obviously, $I_B(M,x,x') = (0,0)$ for all $x' \in M_0^B$, as desired. For $x' \in M_{01}^B$ we know that the first element of $I_B(M,x,x')$ is 0 by definition of $M_{01}^B$. The second element of $I_B(M,x,x')$ is $x - p_B^{\min}$ by point 3, and $p_B^{\min}$ is at least $x$ by point 2, thus $x - p_B^{\min}$ is non-positive and $I_B(M,x,x) \ge I_B(M,x,x')$ for all $x' \in M_{01}^B$, as desired. For $x' \in M_1^B$ we have $I_B(M,x,x') = (x - p_B^{\max}, x - p_B^{\min})$, and both these expressions are negative as $p_B^{\max}$ and $p_B^{\min}$ are both at least $x$. Hence $I_B(M,x,x) \ge I_B(M,x,x')$ for all $x' \in M_1^B$, as desired. This establishes that $I_B(M,x,x) \ge I_B(M,x,x')$ for all $x' \in D_B$ and $x \in M_0^B$.

    Next, let $x$ be a bid in $M_{01}^B$, so that $I_B(M,x,x) = (0,x - p_B^{\min})$. For $x' \in M_0^B$ we have that $I_B(M,x,x') = (0,0) \le I_B(M,x,x)$. For $x' \in M_{01}^B$ we have that $I_B(M,x,x) = I_B(M,x,x')$ by definition of $M_{01}^B$ and the properties of $p_B^{\min}$ stated under point 3 of the claim. For $x' \in M_1^B$ we have $I_B(M,x,x') = (x - p_B^{\max}, x - p_B^{\min})$. Since $p_B^{\max} \ge x$ the expression $x - p_B^{\max}$ is non-positive, so also in this case we have $I_B(M,x,x) \ge I_B(M,x,x')$. This establishes that $I_B(M,x,x) \ge I_B(M,x,x')$ for all $x' \in D_B$ and $x \in M_{01}^B$.

    Lastly let $x$ be a bid in $M_1^B$, so that $I_B(M,x,x) = (x - p_B^{\max}, x - p_B^{\min})$.  For any bid $x' \in M_0^B \cup M_{01}^B$ we have that trade may not occur, so that the first element of $I_B(M,x,x')$ is 0, which is less than or equal to $x - p_B^{\max}$. If $x' \in M_0^B$ then also the second element of $I_B(M,x,x')$ is 0, again less than $x - p_B^{\min}$. If $x' \in M_{01}^B$ then the second element of $I_B(M,x,x')$ is $x - p_B^{\min}$, equal to the second element of $I_B(M,x,x)$. Thus $I_B(M,x,x) \ge I(M,x,x')$ for all $x' \in M_0^B \cup M_{01}^B$. For a bid $x' \in M_1^B$ it simply holds that $I_B(M,x,x') = (x - p_B^{\max}, x - p_B^{\min}) = I_B(M,x,x)$, by the properties of $p_B^{\max}$ and $p_B^{\min}$ in point 3 of the claim. This establishes that $I_B(M,x,x) \ge I_B(M,x,x')$ for all $x' \in D_B$ and $x \in M_1^B$. Altogether this yields that $I_B(M,x,x) \ge I_B(M,x,x')$ for all $x,x' \in D_B$, i.e., $M$ is NOM for the buyer.
\end{proof}

\section{Conclusions}


This paper introduces a framework for designing NOM mechanisms with payments.
We formalise the added flexibility of such mechanisms in the form of labellings which allows us to characterise the set of implementable allocation functions.
We apply our analysis to bilateral trade and show a surprising dichotomy between BNOM and WNOM with regards to budget balance -- we provide an efficient, IR, $\alpha$-WBB WNOM mechanism and show that no BNOM mechanism exists with bounded $\alpha$.
In light of our results we believe that NOM mechanisms deserve to be better understood, and reasoning about the labellings inherent to such mechanisms appears to be a useful way to do so.
There are a number of settings in the mechanism design literature to which our framework could be applied next not covered in \cite{Troyan2020}, including, for example, procurement auctions, combinatorial auctions, and two-sided markets.
Following \cite{Aziz2020} it would of course be useful to explore computational aspects to these mechanisms and further how the two sides of NOM relate to one another.
Our general characterisation can also be helpful in better understanding multidimensional domains, for which strategyproofness is known be restrictive \cite{Koutsoupias2021}.

\bibliographystyle{apalike} 
\bibliography{references}

\end{document}